\documentclass[journal]{IEEEtran}
\usepackage{xcolor}
\usepackage{amsmath,amsthm,amsfonts}
\usepackage{graphicx}
\usepackage{epsf}
\usepackage{mdwmath}
\usepackage{mdwtab}
\usepackage{cite}
\usepackage{amssymb}
\usepackage{algorithmic}
\usepackage{algorithm}

\usepackage{amssymb}
\usepackage{algorithmic}
\usepackage{algorithm}
\usepackage{bm}
\usepackage{flushend}

\newtheorem{thm}{Theorem}
\newtheorem{cor}{Corollary}

\newtheorem{definition}{Definition}

\theoremstyle{plain}

\newcommand{\E}{{\cal{E}}}

\newcommand{\y}{\bm{y}}

\newcommand{\gL}{\gamma_{\mathrm{L}}}
\newcommand{\gU}{\gamma_{\mathrm{U}}}

\def\E{\mathbb{E}}
\def\P{\mathbb{P}}
\def\1{\bm{1}}

\def\S{\mathcal{S}}

\begin{document}

\title{Sequential Testing for Sparse Recovery}

\author{\IEEEauthorblockN{Matthew L. Malloy, \emph{Member, IEEE}, Robert D. Nowak, \emph{Fellow, IEEE}}
\thanks{
\newline \indent Manuscript received December 7, 2012; revised January 15, 2013; accepted September 27, 2014.
This work was supported in part by AFOSR grant number FA9550-09-1-0140 and FA9550-09-1-0643.
Portions of this work were presented at the Asilomar Conference on Signals, Systems, and Computers, November 2011 \cite{ malloy2011limits}, and the International Symposium on Information Theory in Saint Petersburg, August 2011 \cite{malloy2011sequential}.  
\newline \indent The authors would like to thank Cun-Hui Zhang, Rui Castro, and Jarvis Haupt for helpful discussions.\newline \indent M. L. Malloy and R. D. Nowak are with the Department of Electrical and Computer Engineering, University of Wisconsin, Madison, WI 53715 USA (e-mail: mmalloy@wisc.edu, nowak@engr.wisc.edu).
\newline \indent Copyright (c) 2014 IEEE. Personal use of this material is permitted.  However, permission to use this material for any other purposes must be obtained from the IEEE by sending a request to pubs-permissions@ieee.org.
}
}

\maketitle

\begin{abstract}
This paper studies sequential methods for recovery of sparse signals in high dimensions.  When compared to fixed sample size procedures, in the sparse setting, sequential methods can result in a large reduction in the number of samples needed for reliable signal support recovery.  Starting with a lower bound, we show any \emph{coordinate-wise} sequential sampling procedure fails in the high dimensional limit provided the average number of measurements per dimension is less then {$\log ( s ) / D(P_0||P_1)$,} where $s$ is the level of sparsity and $D(P_0||P_1)$ the Kullback–-Leibler divergence between the underlying distributions.  
A series of Sequential Probability Ratio Tests (SPRT) which require complete knowledge of the underlying distributions is shown to achieve this bound.    Motivated by real world experiments and recent work in adaptive sensing, we introduce a simple procedure termed \emph{Sequential Thresholding} which {can be implemented when the underlying testing problem satisfies a monotone likelihood ratio assumption}.  Sequential Thresholding guarantees exact support recovery provided the average number of measurements per dimension grows faster than {$\log ( s )  / D(P_0||P_1)$}, achieving the lower bound.  For comparison, we show any \emph{non-sequential} procedure fails provided the number of measurements grows at a rate less than {$\log (n) / D(P_1||P_0)$}, where $n$ is the total dimension of the problem. 
\end{abstract}

\section{Introduction}
Signal support recovery in high dimensions is a fundamental problem arising in many aspects of science and engineering.  The goal of the basic problem is to determine, based on noisy observations, a sparse set of elements that somehow differ from the others.  

In this paper we study the following problem.  {Consider a support set $\S \subset \{1,\dots,n\}$ and a random variable $Y_{i} \in \mathcal{Y}$ distributed according to 
\begin{eqnarray} \label{eqn:underlyingstats}
    Y_{i} \sim \left\{
           \begin{array}{ll}
             P_0 & i \not \in \mathcal{S} \\
             P_1 & i \in \mathcal{S}
           \end{array}
         \right. \qquad i = 1,\dots,n
\end{eqnarray}
where $P_0$ and $P_1$ are probability densities or mass functions with respect to a common dominating measure. }
The dimension of the problem, $n$, is large --  perhaps thousands or millions or more -- but the support set $\S$ is sparse in the sense that the number of elements following distribution $P_1$ is much less than the dimension, i.e., $|\S| = s \ll n$. The goal of the sparse recovery problem is to identify the set $\S$ from multiple independent realizations of the random variables $Y_1, Y_2,\dots,Y_n$.

The conventional theoretical treatment of this problem assumes that the samples are collected prior to data analysis in what is refereed to as a \emph{non-sequential} (or \emph{fixed sample size}) setting.  In this case, $m$ samples of each component are made ($m$ samples of $Y_{i}$ are gathered for each index $i$) and any test for inclusion in $\S$ is performed after the data is collected. The fundamental limits of reliable recovery are readily characterized in terms of Kullback-Leibler divergence and dimension (see Sec. \ref{sec:NonSeq}).

On the other hand, information gathering systems encountered in practice are often tasked with measuring some temporal signal or process, leaving the potential for the system to \emph{adapt} the sampling approach based on prior observations.  In this \emph{sequential} setting, the decision to take an additional sample of any component $i$ is based on prior realizations of that component.  Herein lies the advantage of sequential methods: if prior samples indicate a particular component belongs (or doesn't belong) to $\S$ with sufficient certainty, measurement of that component can cease, and resources can be diverted to a more uncertain element.  The focus of this paper is on the fundamental limits of recovery of such sequential systems.

\subsection{Main Contributions}
The results presented in this paper are in terms of asymptotic rate at which the average number of samples per dimension, denoted $m$, must increase with $n$ to ensure exact recovery of $\S$ for any fixed distributions $P_0$ and $P_1$.   For a given procedure, the probability of correctly recovering the set $\mathcal{S}$ depends on the triple $(n,s,m)$.
As the dimension of the problem grows (as $n \rightarrow \infty$), correctly recovering $\mathcal{S}$ becomes increasingly difficult, and the number of measurements must also increase if we hope to recover $\mathcal{S}$.   One manner in which we can quantify the performance of a procedure is the rate at which $m$ must grow as a function of $n$ and $s$ to ensure recovery of $\mathcal{S}$. 

As such, the main contributions are \emph{1)} to derive a lower bound on the number of measurements required for success of any {coordinate-wise} sequential procedure in the sparse setting, {\emph{2)} introduce a simple sequential procedure termed \emph{Sequential Thresholding}  which can often be implemented when $P_1$ is not fully specified (more specifically, when the underlying testing problem satisfies a monotone likelihood ratio assumption -- see Sec. \ref{sec:ImpST} and Def. \ref{def:MLR} for details)} and show this simple procedure is asymptotically optimal, \emph{3)} compare this procedure to the known optimal SPRT, and lastly \emph{4)} compare these results to the performance of any non-sequential procedure.   Table \ref{table1} summarizes these results.  

\begin{table}[ht!]
\caption{Average number of measurements per dimension for exact support recovery in high dimensional limit} 
\centering
\begin{tabular}{p{1.8cm}|p{1.95cm}|p{0.95cm}|p{2.35cm}}
 \hline
  \hline
  \emph{ non-sequential}            &  $m \geq \frac{\log n }{D(P_1||P_0)}$ & necessary  &  \vspace{.2cm} \\
  \hline
  \emph{ sequential}  &  $m \geq \frac{\log s }{D(P_0||P_1)}$ & necessary & \vspace{.2cm}\\
  \hline  
\emph{SPRT based \newline procedure} &   $m > \frac{\log s }{D(P_0||P_1)} $ & sufficient   & requires exact kno-\newline wledge of $P_0$, $P_1$ 
\\  
  \hline
  \emph{Sequential \newline Thresholding}   &  $m > \frac{\log s}{D(P_0||P_1)} $ & sufficient & does not require exact knowledge of $P_1$
  \\
  \hline
   
\end{tabular}
\label{table1}
\end{table}

These developments are intriguing primarily for two reasons.  First, the results show that \emph{sequential} procedures succeed when the number of measurements per dimension increases at a rate logarithmic in the level of \emph{sparsity}, i.e. $\log s$.  In contrast, well known results from statistical testing show \emph{non-sequential} procedures require the average number of measurements per dimension to increase at a rate logarithmic in the \emph{dimension}, i.e. $\log n$.   
Secondly, \emph{Sequential Thresholding}, a simple, practical procedure introduced here, achieves optimal performance as the dimension grows large.   
{The procedure operates by repeatedly discarding from consideration components that exhibit strong evidence of following $P_0$.  Sequential Thresholding can be implemented when the sparse components follow certain one-sided composite hypotheses -- specifically, Sequential Thresholding requires full knowledge of $P_0$ and knowledge of a test statistic that satisfies a monotonic likelihood ratio assumption (see Def. \ref{def:MLR}).   }

\subsection{Motivation}
The problem of sparse signal recovery using sequential measurements arises in a number of commonly encountered problems in science and engineering.  In communications, spectrum sensing for cognitive radio aims to identify unoccupied communication bands in the electromagnetic spectrum.  Most bands will be occupied by primary users, but these users may come and go, leaving a {sparse set of bands} momentarily open and available for use by secondary transmitters.  As noisy samples of these occupied and unoccupied bands are collected in a temporal manner, sequential methods are a natural fit to map the occupation of the spectrum; in fact, recent work in spectrum sensing has given considerable attention to such approaches (see, for example \cite{Castro}, \cite{5454086}). 

Another captivating example a of sparse recovery problem where sequential methods are currently employed is that of the \emph{Search for Extraterrestrial Intelligence} (SETI) project.  Researchers at the SETI institute sense for narrowband electromagnetic energy from distant star systems using large antenna arrays, with the hopes of finding extraterrestrial transmission.  The dimension of the problem consists of over 100 billion stars in the Milky Way alone, each with 9 million potential `frequencies' in which to sense for narrow band energy.  The subset of planetary systems with extraterrestrial transmission is sparse (since, to the best of our knowledge, SETI is yet to make a contact).  Moreover, while researchers may have a good idea of the distribution of the background noise, $P_0$, complete knowledge of $P_1$ is of course not available, making procedures based on sequential probability ratio testing impractical.  Roughly speaking, researchers at SETI use a sequential procedure that repeatedly tests energy levels against a threshold up to five times \cite{SETI, NYtimes}.  If any of the up to five measurements are below the threshold, the procedure passes to the next frequency/star.  Should the measurements exceed the threshold on all five occasions, measurements of that star and frequency are passed to an operator for further inspection. This procedure is closely related to Sequential Thresholding. Sequential Thresholding results in substantial gains over fixed sample size procedures and can be computed without full knowledge of $P_1$.  

Sparse recovery also underlies a number of recent assay studies in biology.  {Here, biologists estimate a sparse set of genes or proteins that are critically involved in a certain process or function.  As an example, the study in \cite{hsw-drsih-2008} aims to identify a small number of genes (approximately $s =100$ out of  $n=13,071$ total) that are important to virus replication in fruit files cells.  The involvement of each gene is measured as follows.  First the functionality of the gene is suppressed (using a single gene \emph{knockout}) and the fruit fly cells are exposed to the virus under study.  Associated with the virus is a fluorescent marker, and the virus's ability to replicate is quantified by measuring the florescence produced by the infected cells.  
In our model, the level of florescence observed when gene $i$ is suppressed corresponds to a realization of the random variable $Y_i$.   The biologists may have good estimates of the null distribution, $P_0$, but not of the alternative distribution, $P_1$, again making procedures based on the SPRT difficult to implement.  A number of recent publications have implemented various multi-stage (thus sequential) procedures \cite{Muller2007, sata3, Zehetmayer, hsw-drsih-2008} that operate without full knowledge of $P_1$.  The proposed procedures in general aim to reduce the total dimension of the problem and then employ traditional recovery techniques.  While a number of authors suspect such sequential methods result in increased sensitivity, the gains are not fully theoretically quantified. 
}

\subsection{Related Work}  \label{sec:IntroPrior}
{Many of the fundamental results in sequential analysis were developed by Wald, and formalized in his book, Sequential Analysis \cite{wald}.  The sequential probability ratio test (SPRT) was shown to be optimal in terms of minimizing the error probabilities and expected number of measurements for a simple binary hypothesis test.  A handful of issues arise when exact knowledge of the distributions is unavailable, including loss of optimality, which can make the SPRT impractical in many scenarios.  More specifically, in a parametric setting with a monotone likelihood ratio, if the SPRT is implemented with thresholds based on an incorrect parameter, the test can result in arbitrarily large sample size  (see \cite{lai2000sequential}, \cite{ghosh1991handbook},   and  Sec. \ref{sec:SPRTimp}). }

{Aimed at addressing the deficiencies of the SPRT, a vast body of literature is devoted to sequential tests of composite hypothesis (including work by Wald \cite{1945Wald}), and the more restrictive case of a monotone likelihood ratio setting.  Before addressing sequential testing in the sparse and high dimensional setting, we give a brief overview of this literature and refer the reader to \cite{ghosh1970sequential}, \cite{ghosh1991handbook}, and \cite{lai2000sequential} for a more complete summary of sequential testing for composite hypothesis testing. One of the most prevalent sequential procedures applicable to the monotone likelihood ratio setting is that of Lorden \cite{lorden19762}, who proposed a procedure termed a 2-SPRT, which was shown to be optimal for single parameter exponential families in a sense proposed by Kiefer and Weiss \cite{kiefer1957some} (in which the error rates are minimized at a particular worst case parameter point, but not universally). Schwarz \cite{schwarz1962asymptotic} addressed the insufficiencies of the SPRT by examining the shape of the decision regions corresponding to asymptotically optimal tests (in a Bayesian sense) as the sample size grows large.  The resulting procedure in brief operates as a sequential generalized likelihood ratio test with an \emph{indifference zone}. 
Many further approaches have been suggested in the monotone likelihood setting:   including linear stopping boundaries \cite{anderson1960modification}, curved stopping boundaries \cite{SeqAnalysis}, and various forms of truncated SPRTs.  Of these, the work of \cite{lai1988nearly} gave nearly optimal results over a  wide class of problems, but was limited practically in that the curved stopping boundary is defined asymptotically.  Ultimately, unlike the case of simple hypothesis, practical procedures with \emph{universally optimal} theoretical guarantees are not available. }


{
Sequential testing for sparse signals was perhaps first studied by Posner in \cite{1057838}.  Motivated by the the problem of finding a lost satellite in the sky, Posner aimed to minimize the expected search time using a multistage procedure.  
Posner's procedure is closely related to the high dimensional extension of the sequential probability ratio test (see Sec. \ref{sec:SPRT} for details).    Sequential approaches to the high dimensional sparse recovery problem have recently been given increased attention, perhaps motivated by the success of exploiting sparsity in other areas.  Related work includes \cite{4585341, 5710433}, in which the authors extend the work of \cite{1057838} to include multiple targets, encompassing a more general model, {and the work of \cite{5961845}, which aims to find a rare element amongst infinitely many}. 
}

In some of the first work to quantify the gains of sequential methods for sparse recovery \cite{5074721, Haupt}, the authors proposed a sequential procedure for recovery in additive Gaussian noise, termed \emph{Distilled Sensing}.  Our \emph{Sequential Thresholding} approach is similar to the Distilled Sensing method, however, there are a number of distinctions.   In this work we are concerned with the probability of error in exact recovery of the sparse support $\S$; Distilled Sensing controls the false discovery and non-discovery rates which is less demanding than control of the family-wise error rate.  Controlling a distinct metric gives rise to significant algorithmic differences.  From an algorithmic perspective both procedures involve of a number of passes each of which discards components following $P_0$. For theoretical guarantees, Distilled Sensing requires $\sim\log \log n$ passes (see Theorem III.1 of \cite{Haupt}), followed by an estimation step, while Sequential Thresholding uses $\sim \log n$ passes, without a final estimation step.  From an analysis perspective, the results in this paper are applicable to a larger class of problems characterized by finite Kullback-Leibler divergence; the Distilled Sensing approach is specific to the Gaussian setting.  Lastly, the results in \cite{Haupt} are presented in terms of a parametric scaling of the sparsity with dimension, while no such scaling is assumed here.

{Also closely related to the work here are the lower bounds of \cite{castro2012adaptive}. The lower bounds presented in \cite{castro2012adaptive} are stronger in that they are not restricted to the \emph{coordinate-wise} assumption, but weaker in that they are terms of the expected set difference and restricted to the Gaussian setting.  The results of \cite{castro2012adaptive} were published after the initial work in \cite{malloy2011limits, malloy2011sequential}.}


Another related set of problems is that of finding the \emph{best arm} in a multi-armed bandit game \cite{bubeck2009pure,  jamieson2013finding, even2002pac}. { Some approaches to this problem are similar in nature to Sequential Thresholding, namely the \emph{median elimination} procedure of \cite{even2002pac}, but the problem setting is fundamentally different in that the procedure aims to return a single element that is \emph{approximately} best.  Other work in the \emph{best arm} literature focuses on finding a single sparse element with \emph{tests of uniformly small probability of error} \cite{jamieson2013finding, bubeck2009pure}, which is distinct from the setting studied here. 
}
 


\subsection{Organization}

The remainder of the paper is organized as follows.  In Sec. \ref{sec:ProbForm} we formalize the problem.  Sec. \ref{Sec:limits} derives the necessary condition on the number of samples required for exact recovery using any procedure.  For comparison, Sec. \ref{sec:NonSeq} derives a necessary condition on the average number of measurements for non-sequential procedures.  Next, Sec. \ref{sec:SPRT} analyzes the SPRT in the sparse setting and discusses some of the shortcomings of the test when exact knowledge of the distributions is not available.  Lastly, Sec. \ref{sec:SeqThres} introduces \emph{Sequential Thresholding} and analyzes its performance.


\section{Problem Formulation} \label{sec:ProbForm}
Let $\S$ be a subset of $\{1,\dots,n\}$ with cardinality $s = |\S|$.  {For any index $i \in \{1,\dots,n\}$, the random variables  $Y_{i} \in \mathcal{Y}$ are independent and distributed according to (\ref{eqn:underlyingstats}),
where $P_0$ and $P_1$ are probability distributions or mass functions with common support\footnote{{The assumption of common support can be relaxed in practice and likely leads to substantial gains, although this is not investigated.}} on $\mathcal{Y}$ defined with respect to a common dominating measure.   In words, the random variable $Y_{i}$ follows distribution $P_1(\cdot)$ if $i$ belongs to $\S$, and follows $P_0(\cdot)$ otherwise. We write $Y_{i,j}$,  for $j=1,2, \dots$, to index multiple i.i.d. samples of $Y_i$, and we refer to $P_0$ as the null distribution, and $P_1$ the alternative.
}
{
Our analysis is concerned with exact recovery of the set $\S$.     The family wise error rate is defined as the probability that the estimated support set differs from the true support set:
\begin{eqnarray} \label{eqn:FWER1}
 \P_e = \P(\hat{\S} \neq \S).
\end{eqnarray}
}


 
The log-likelihood ratio statistic comprised of multiple i.i.d. samples of a particular index is defined as:
\begin{eqnarray} \label{eqn:LR}
  L_i^{(\ell)}\left(Y_{i,1},\dots,Y_{i,\ell} \right) :=  \sum_{j=1}^{\ell} \log \frac{P_1(Y_{i,j})}{P_0(Y_{i,j})}.
\end{eqnarray}
Here, the superscript $\ell$ explicitly indicates the number of samples used to form the likelihood ratio and is suppressed when unambiguous.   
The Kullback-Leibler divergence from distribution $P_1$ to $P_0$ is defined as:
\begin{eqnarray} \nonumber
  D(P_1||P_0) = \E_1\left[ \log \frac{P_1(Y)}{P_0(Y)} \right]
\end{eqnarray}
where $\E_1\left[\cdot \right]$ is expectation with respect to distribution $P_1$, which gives the usual convergence of the normalized likelihood ratio as $\ell$ grows large:
\begin{eqnarray}  \label{eqn:DivConv}
\frac{1}{\ell}L_i^{(\ell)} \overset{a.s.} \longrightarrow \left\{
           \begin{array}{ll}
             -D(P_0||P_1)  & i \not \in \mathcal{S} \\
             D(P_1||P_0) & i \in \mathcal{S}.
           \end{array}
         \right.
\end{eqnarray}
It is sometimes convenient to state results in terms of the maximum of $D(P_0||P_1)$ and $D(P_1||P_0)$.  In this case, we define 
\begin{eqnarray} \nonumber
D_{\mathrm{KL}} = \max \left\{D(P_0||P_1), D(P_1||P_0) \right\}.
\end{eqnarray}
In order to bound rates of convergence of particular testing procedures, we make use of the variance of the log-likelihood ratio, denoted
\begin{eqnarray} \nonumber
\sigma^2(P_1||P_0) &=& \mathrm{var}\left(L_i^{(1)} | i \in S \right)   \\
& =& \E_1 \left[ \left( \log \frac{P_1(Y)}{P_0(Y)}  - D(P_1||P_0) \right)^2  \right]. \nonumber
\end{eqnarray}

A sampling procedure $\Gamma$ is a method used to determine the number of samples taken of each index.  To be precise in characterizing a sampling procedure, we present four definitions.
\begin{definition}
\emph{Sampling procedure.} {
A collection of functions $\Gamma_{i,j}: \{Y_{i,1},\dots,Y_{i,j-1}\} \mapsto \{0,1\}$, for $i \in \{1,\dots,n\}$ and $j \in \mathbb{N}$ that defines the number of samples of $Y_i$ that are observed. } Specifically, if $\Gamma_{i,j} = 1$, then $Y_{i,j}$ is observed, and can be used in estimation of $\mathcal{S}$.  Conversely, if $\Gamma_{i,j} = 0$, then $Y_{i,j}$ is not observed, and is not used in estimation of $\mathcal{S}$.
\end{definition}


\begin{definition} \label{def:NS}
\emph{Non-sequential (fixed sample size) sampling procedure.}
Any sampling procedure such that $\Gamma_{i,j}$ is not a function of $Y_{i',j'}$ for any $i',j'$.
\end{definition} 
\noindent

\begin{definition} \label{def:S}
\emph{Sequential sampling procedure.}
A sampling procedure in which $\Gamma_{i,j}$ is allowed to depend on previous samples, specifically, $\Gamma_{i,j}: \{Y_{i,1},\dots,Y_{i,j-1}\} \mapsto \{0,1\}$.
\end{definition}

\begin{definition} \label{def:CW}
\emph{Uniform coordinate-wise sampling procedure.}
{A sampling procedure in which $\Gamma_{i,j}$ is not a function of $i$.}
\end{definition}

Sequential procedures can make use of information as it becomes available to adjust the sample size, while non-sequential procedures, or \emph{fixed sample size} procedures, fix the number of samples taken \emph{a priori}.  Note that under this definition, the set of non-sequential procedures are a subset of sequential procedures.

{ 
In the lower bounds developed in this paper our consideration is limited to procedures that test each index in an identical manner (see definition \ref{def:CW}). The uniform coordinate-wise assumption also implies the procedure only uses samples of component $i$ to make inference about that particular component.  More specifically, the decision to re-measure a particular component or include it in the estimate of $\S$ depends only on samples of that component.  As the dimension of the problem grows large (which is our regime of interest), there is no loss of optimality associated with this restriction\footnote{The lower bounds in \cite{castro2012adaptive}, which are restricted to the Gaussian setting but do not make a coordinate-wise assumption, match the scaling of the lower bounds presented here.}.} 

In order to make a fair comparison between different procedures, we limit the total number of samples in expectation.  For any procedure we require
\begin{eqnarray} \label{eqn:budget}
 \mathbb{E}\left[\sum_{i,j} \Gamma_{i,j} \right]\leq nm
\end{eqnarray}
for some $m \geq 0$.   This simply implies, on average, the procedure uses $m$ or fewer samples per dimension.

The family wise error rate of any procedure used to estimate $\mathcal{S}$ depends on the underlying distributions $P_0$ and $P_1$, the dimension, $n$, the level of sparsity $s$, and the average number of samples per component, $m$. 
Throughout, $s$ and $m$ are non-decreasing functions in $n$ (and thus, the set $\S$ is also a function of $n$).  We suppress this dependence on $n$ for ease of exposition.  Our focus will be on finding the relationship between the triple $(n,s,m)$  such that for any fixed distributions $P_0$ and $P_1$, either $\lim_{n\rightarrow \infty} \P_e =0$ (the procedure is reliable) or $\lim_{n\rightarrow \infty} \P_e > 0$ (the procedure is unreliable).
We assume $s \leq n/2$ (without loss of generality provided $P_0$ and $P_1$ are known, since if $s > n/2$, one can re-label the problem, swapping $P_0$ and $P_1$). As we are interested in sparse problems, some of the results require the assumption that $\lim_{n \rightarrow \infty} \frac{s}{n} = 0$, which is termed \emph{sub-linear} sparsity, but this scaling is stated explicitly when needed.

\section{Limits of Reliable Recovery} 
This section presents lower bounds on the number of measurements required for reliable recovery by any procedure in both the sequential and non-sequential setting.  The bounds are in terms of the expected number of samples per dimension.

\subsection{Limitation of Sequential Procedures} \label{Sec:limits}
The following theorem quantifies the limitations of \emph{any} procedure, which includes both sequential and non-sequential procedures, as non-sequential procedures are a subset of sequential procedures (from Def. \ref{def:NS} and Def. \ref{def:S}).  The bound applies to finite problems, but also implies a necessary rate at which $m$ must grow with $n$ for reliable recovery, captured in the ensuing corollary.  


\begin{thm} \label{thm:LBSeq} \emph{Finite sample limitations of sequential procedures.}
Any uniform {coordinate-wise} (sequential) procedure with
\begin{eqnarray} \nonumber
m \leq \frac{ \log s + \log \left(\frac{1}{4\delta}\right)}{D_{\mathrm{KL}}}
\end{eqnarray}
also has 
\begin{eqnarray} \nonumber
\P_e &\geq & 1 - e^{-\delta } \approx \delta
\end{eqnarray}
where the approximation holds for small $\delta$.
\end{thm}

\begin{proof}
See Appendix \ref{app:LB}.
\end{proof}

Thm. \ref{thm:LBSeq} establishes a lower bound on the expected number of samples needed to achieve a particular family wise error rate.  As the dimension of the problem grows, it provides us with a necessary condition for reliable recovery.  

\begin{cor} \label{STcor:LB} \emph{Limitations of sequential procedures.}
Assume $\lim_{n\rightarrow \infty} s/n = 0$. Any {uniform coordinate-wise} (sequential) procedure with
\begin{eqnarray} \nonumber
{\limsup_{n \rightarrow \infty} }\; \frac{m}{\log s} \leq \frac{1}{D(P_0||P_1)}
\end{eqnarray}
also has ${\liminf_{n\rightarrow \infty}} \P_e > 1/5$.  
\end{cor}
\begin{proof}
Thm. \ref{thm:LBSeq} implies that if $m \leq \frac{\log s}{D_{\mathrm{KL}} }$ then $\P_e \geq 1 - e^{-1/4} > 1/5$.  Dividing by $\log s $ and taking the limit as $n \rightarrow \infty$ would give the lemma if $D_{\mathrm{KL}}  =D(P_0||P_1)$.  Instead, inspecting the proof of Thm. \ref{thm:LBSeq}, it is easily verified that if $\lim_{n\rightarrow \infty} s/n = 0$, the analysis follows with $D_{\mathrm{KL}}$  replaced by $D(P_0||P_1)$. 
\end{proof}
In words, if the number of samples per dimension grows at a rate slower than logarithmically in the level of sparsity, no procedure can reliably recover $\S$. In shorthand notation, if $m \leq \frac{\log s}{D(P_0||P_1)}$ then $\P_e$ can not be driven to zero, and recovery of $\S$ is unreliable in the large $n$ limit.  


\subsection{Limitation of Non-Sequential Procedures }
\label{sec:NonSeq}

Non-sequential methods, which sample each index a fixed number of times, can require significantly more measurements than sequential procedures.  In the following theorem we state a necessary condition on $m$ for any reliable  non-sequential procedure.  The proof is based on analysis of the  \emph{Chernoff Information} \cite{Cover:1991:EIT:129837}.
Our consideration is restricted to {non-sequential coordinate-wise procedures (which by definition sample each component $i = 1,\dots,n$ exactly $m$ times).}

\begin{thm} \emph{Limitation of non-sequential procedures.} \label{thm:NSasmp}
Assume $\lim_{n\rightarrow \infty} s/n = 0$. Any non-sequential {uniform coordinate-wise} procedure with
\begin{eqnarray} \label{eqn:asmpNS}
  {\limsup_{n\rightarrow \infty}} \ \frac{m}{\log n} < \frac{1}{D(P_1||P_0)}
\end{eqnarray}
also has
\begin{eqnarray} \nonumber
{\liminf_{n\rightarrow \infty} \ \mathbb{P}_e \geq 1/2.}
\end{eqnarray}
\end{thm}
\begin{proof}
See Appendix \ref{app:NSMethods}.
\end{proof}


\section{Sequential Probability Ratio Testing}\label{sec:SPRT}
\subsection{The SPRT}

Provided $P_0$ and $P_1$ are known, sequential probability ratio tests are optimal for binary hypothesis tests in terms of minimizing the expected number of measurements for any error probabilities (shown originally in \cite{1948}); this optimality translates to the high dimensional case by simply considering $n$ parallel SPRTs.  

Each individual SPRT operates by continuing to measure a component if the corresponding likelihood ratio is within an upper and lower threshold, and terminating measurement otherwise.  For scalar thresholds $\gL$ and $\gU$, the procedure is defined as
\begin{eqnarray} \label{eqn:SPRT:RULES}
  \Gamma_{i,j'+1} &=& \left\{
                      \begin{array}{ll}
                        1 &\mathrm{if} \quad  \gL \leq \prod_{j=1}^{j'} \frac{P_1(Y_{i,j} )}{P_0(Y_{i,j})} \leq \gU \\
                        0 &\mathrm{else}
                      \end{array}
                    \right. 
\end{eqnarray}
where $\prod_{j=1}^{j'} \frac{P_1(Y_{i,j} )}{P_0(Y_{i,j})}$ is the likelihood ratio comprised of all prior samples.  If the likelihood ratio falls below $\gL$, the SPRT labels index $i$ as \emph{not} belonging to $\hat{\S}$; if the likelihood ratio exceeds $ \gU$, index $i$ is assigned to $\hat{\S}$.  Equivalently, the test can be implemented in the $\log$-likelihood domain, and $L_i^{(j')}$ can be compared against $\log (\gL)$ and $\log(\gU)$.  The procedure requires a random number of samples of each component,  denoted $J_i$, and defined as
\begin{eqnarray} \nonumber
J_i := \min \{j : \Gamma_{i,j+1} = 0\}.   
\end{eqnarray}

As we proceed we make a minor assumption on the distribution of the log-likelihood statistic. Specifically, the ensuing theorem and proof require existence of positive constants $C_1$ and $C_2$ such that 
\begin{eqnarray}   \nonumber
&& \mathbb{E}[L_i^{(J_i)} | L_i^{(J_i)} < \log \gL ] \geq \log \gL - C_1 \\
&&\mathbb{E} [L_i ^{(J_i)} | L_i^{(J_i)} > \log \gU] \leq \log \gU + C_2\label{STeqn:ass1}
\end{eqnarray}
for all thresholds $\gU$ and $\gL$.   In some cases, bounds for $C_1$ and $C_2$ are known (see \cite{1945Wald}, p.145, where explicit expressions for the Bernoulli and Gaussian case are given).   In words, the requirement is the existence of a constant that bounds the expected value of the $\log$-likelihood ratio when the procedure terminates, regardless of the value of the threshold.  
 This condition is satisfied when $L_i^{(1)}$ follows any bounded distribution, Gaussian distributions, exponential distributions, among others.  It is not satisfied by distributions with infinite variance or polynomial tails.  A more thorough discussion of this restriction is studied in \cite{SPRTbounds1960}.

\begin{thm} \emph{Ability of the SPRT.} \label{thm:SPRTach}
The SPRT procedure with thresholds $\gL = \frac{1}{s^{1+\epsilon}}$ and $\gU = {(n-s)}^{1+\epsilon}$, any $\epsilon >0$, has 
\begin{eqnarray} \nonumber
\lim_{n\rightarrow \infty} \P_e = 0 
\end{eqnarray}
and 
\begin{eqnarray} \nonumber
  \lim_{n\rightarrow \infty} \frac{m}{\log s } \leq \frac{1+\epsilon}{D(P_0||P_1)}.
\end{eqnarray}
provided $s< n /\log n$, and the condition in (\ref{STeqn:ass1}) is satisfied.  
\end{thm}
\begin{proof}
See Appendix \ref{app:SPRT}.
\end{proof}

\subsection{Implementation Issues} \label{sec:SPRTimp} Implementing an SPRT on each component can be challenging for many problems encountered in practice.  While the SPRT is optimal when both $P_0$ and $P_1$ are known and testing a single component amounts to a simple binary hypothesis test, scenarios often arise where some parameter of distribution $P_1$ is unknown.  When some parameter of $P_1$ is unknown, the likelihood ratio cannot be formed, and sufficient statistics for the likelihood ratio result in adjustments to the thresholds based on the unknown parameters of distribution $P_1$.  With incorrect thresholds, the SPRT is no longer optimal.  To see this more concretely, consider a problem where $P_1$ is a normal distribution with an unknown positive mean $\mu$ and unit variance, and $P_0$ is a zero mean standard normal distribution.  Here, the SPRT procedure continues to sample a particular index if 
\begin{eqnarray}
\frac{\log \gL}{\mu} + \frac{j' \mu}{ 2}  \leq  \sum_{j=1}^{j'}  Y_{i,j}  \leq \frac{\log \gU}{ \mu} + \frac{j' \mu}{2},
\end{eqnarray}
equivalent to (\ref{eqn:SPRT:RULES}). While the statistic  $\sum_{j=1}^{j'}  Y_{i,j}$ does not depend on the unknown parameter $\mu$, the thresholds do.  If the test is implemented with an incorrect value of $\mu$, it may result in large sample sizes.
This occurs when $\mu$ is overestimated; for illustration, consider a scenario in which the threshold is set erroneously using $\tilde{\mu} = 2 \mu$, where $\mu$ is the  true mean of $P_1$.  If $P_1$ is the true distribution, this test is then equivalent to waiting for an unbiased random walk to cross a constant threshold, resulting in arbitrarily large sample size.  {Conversely, if $\mu$ is underestimated, the false negative rate of the test becomes arbitrarily large.  This issue lead researchers to study the optimality of the SPRT with an \emph{indifference} region \cite{schwarz1962asymptotic}.   }

{To address these and other deficiencies of the SPRT, a number of tests have been proposed for composite tests in the monotone likelihood setting (see Sec. \ref{sec:IntroPrior} and \cite{lorden19762, schwarz1962asymptotic, lai1988nearly}); we note that these procedures may achieve the lower bound as stated in Theorem \ref{thm:LBSeq}  (although we do not investigate this further here).   Motivated by real world experiments in biology, we present a procedure specific to the sparse setting, termed Sequential Thresholding.  Sequential Thresholding can be interpreted as a truncated SPRT without memory.  Sequential Thresholding is simple to implement and analyze, and provides optimal theoretical guarantees for sparse signals.   }


\section{Sequential Thresholding} \label{sec:SeqThres}
Sequential Thresholding is based on simple idea: repeatedly reduce the dimension of the problem by sequentially eliminating elements that exhibit strong evidence they don't belong to $\S$.  Sequential Thresholding consists of a series of $K$ measurement steps, where each step eliminates from consideration a proportion of the components measured on the prior step.  After the last step, the procedure terminates, and the remaining components are taken as the estimate of $\S$.  

To illustrate the main idea behind the procedure, we first introduce a simplified version of Sequential Thresholding and analyze the simplified procedure for a specific problem.  This \emph{simple} Sequential Thresholding, while not achieving asymptotic optimality,  does admit a simple error analysis.  The more general version of Sequential Thresholding, which does achieve optimality and the lower bound of Cor. \ref{STcor:LB}, is presented in the second half of this section.

\subsection{Example of Simple Sequential Thresholding}
To highlight the main idea behind Sequential Thresholding, and the potential performance gains,  consider a problem where $P_0 \sim {\cal{N}}(0,1)$ and $P_1 \sim {\cal{N}}(\theta,1)$ for some $\theta > 0$.  
{The \emph{simple} Sequential Thresholding procedure requires two inputs: \emph{1)} $\delta >0$, which represents the desired error probability, and \emph{2)} an even integer $m \geq 2$ that defines the average number of samples per index, and hence the total budget.}  On the first step the procedure samples all indices $m/2$ times each, for all $i$, requiring $mn/2$ samples.  These $m/2$ samples are summed for each index; let $T_i$ denote this sum.  If $T_i$ is less than zero, that particular index is not sampled on subsequent passes.
This eliminates a proportion (approximately half) of components following the null distribution (since the median of $T_i$ for $i \not \in S$ is zero).  Indices that exceed the threshold, i.e. $\{i : T_i > 0\}$, are sampled on the subsequent step. This process continues for $K\approx \log_2 n$ steps.  At each step, $T_i$ is defined as the sum of the $m/2$ measurements of index $i$ on that step  (we suppress dependence on the step for ease of notation).  After the $K$th step, the procedure terminates, and estimates $\S$ as the set of indices that have not been eliminated from consideration. Roughly speaking, provided $s \ll n$, the procedure reduces the number of samples taken on each step by half as most components follow the null, which is zero mean.  The total number of samples required by the procedure on all steps {is approximately $\frac{m}{2}\left(n + \frac{n}{2} + \frac{n}{4} +\dots \right) \approx mn$ on average, implying approximately $m$ samples per dimension.}

\begin{algorithm}[h]
\caption{Simple Sequential Thresholding}
\begin{algorithmic} \label{alg:simplesds} 
\STATE{input: desired error probability $\delta$, budget $m \geq 2$}
\STATE{initialize: $\S_1 = \{1,\dots,n\}$, $K = \left \lceil \log_2\left(\frac{2n}{\delta} \right)\right \rceil$ steps}
\FOR{$k = 1,\dots,K$}
\FOR{$ i \in \S_{k}$}
\STATE{{\bf measure}: sample $m/2$ times $Y_i$, denote $T_i$ the sum of \hspace*{.6cm}  these samples }
\STATE{{\bf threshold}: $\S_{k+1} \ := \ \{i\in\S_{k} \ :  T_i > 0 \}$}
\ENDFOR
\ENDFOR
\STATE{output: $\S_{K+1}$} 
\end{algorithmic}
\end{algorithm}

{
\begin{thm} \emph{Reliability of Simple Sequential Thresholding.}
Consider the setting above where $P_0 \sim \mathcal{N}(0,1)$ and $P_1 \sim \mathcal{N}(\theta, 1)$. The simple Sequential Thresholding algorithm with input $\delta >0$ satisfies $ \P_e < \delta$
provided
\begin{eqnarray} \nonumber
m >\frac{\log s + \log\log_2 \left( \frac{2n}{\delta} \right) + \log \left(\frac{1}{\delta}\right)}{\theta^2/4}.
\end{eqnarray}
\end{thm}
\begin{proof}
From a union bound, 
\begin{eqnarray} \label{STeqn:nsbeta}
\P_e\leq (n-s) \P\left(i \in \hat{\S}| i \not \in \S\right) + s \; \P\left(i \not \in \hat{S}| i \in \S \right).
\end{eqnarray}
The false positive event occurs when, for $i\not \in \S$, the index survives all $K$ thresholding steps.  Recall $T_i$ denotes the sum of the $m/2$ samples from any particular step.  By the independence across steps, and since the median of $T_i$  for $i \in \S$ is zero,
\begin{eqnarray}  \label{eqn:STalpha11}
\P\left(i \in \hat{\S}| i \not \in \S\right) = \left(\frac{1}{2}\right)^K \leq \frac{\delta}{2n}.
\end{eqnarray}
The false negative event occurs when for some $i \in \S$, $T_i$ falls below zero on any of the $K$ steps. Applying a union bound and Gaussian tail bound, since $T_i \sim \mathcal{N}\left(m\theta/2, m/2 \right)$, we have
\begin{eqnarray} \nonumber
\P\left(i \not \in \hat{\S}| i \in \S \right) &\leq& \frac{K}{2} \exp\left(-m\theta^2/4 \right)\\ \nonumber
 &\leq& \frac{1}{2}\log_2 \left(\frac{2n}{\delta}\right) \exp\left(-\frac{m\theta^2}{4} \right)  \\
 &\leq & \delta/2
 \label{eqn:STbeta11}
\end{eqnarray}
where the last line follows by asserting 
\begin{eqnarray}
m\geq  \frac{\log s + \log \log_2 \frac{2n}{\delta} + \log \frac{1}{\delta}  }{\theta^2/4}.
\end{eqnarray}
 Combining (\ref{STeqn:nsbeta}), (\ref{eqn:STalpha11}) and (\ref{eqn:STbeta11})  gives $\P_e\leq \delta$.
\end{proof}}

{
The simple Sequential Thresholding procedure requires order $\log s + \log \log n$ samples per dimension.}  While the sub-optimal simple version of Sequential Thresholding does not achieve the lower bound of Theorem \ref{thm:LBSeq}, it does out-perform non-sequential procedures. If $m$ is order  $\log s + \log \log n$, the procedure is reliable.  On the other hand, Sec. \ref{sec:NonSeq} shows reliable recovery with non-sequential methods require $m$ to be order $\log n$.  For large $n$ and small $s$, $\log n$ can be significantly larger than $\log s + \log \log n$, implying that the simplified version of Sequential Thresholding, for sufficiently sparse problems, will succeed with fewer samples than any non-sequential procedure.  Note that the simple version of Sequential Thresholding is a uniform coordinate-wise procedure and the lower bound of Theorem \ref{thm:LBSeq} can be reasonably compared.


{
\subsection{Implementation} \label{sec:ImpST}
One of the main attributes of Sequential Thresholding is that it can be implemented in certain scenarios with limited knowledge of distribution $P_1$; namely, when the underlying testing problem satisfies a monotone likelihood ratio assumption \cite{karlin1956theory}. If there exists a test statistic that is a monotonic transform of the likelihood ratio, regardless of any unknown parameters of $P_1$, then the test can be implemented.  This scenario can arise when testing a one-sided composite hypothesis against a simple alternative.
While apparent in the example of simple Sequential Thresholding above, in which two normal distributions are compared, more generally the monotone likelihood ratio assumption arises when the sparse alternative corresponds to a parametric family of densities.  To be precise, consider the following definition.
\begin{definition}{ \label{def:MLR} \emph{Monotone Likelihood Ratio Assumption.}}
$P_0(Y_i)$ is known and $P_1( Y_i ; \theta)$ is defined by a parametric family of distributions with an unknown parameter $\theta \in {\Theta}$.  The family $\{P_1( Y_i ; \theta)\}_{\theta\in\Theta}$ is said to be a monotone likelihood ratio family with respect to $P_0$ in the scalar statistic $T_i^{(\ell)} =T_i^{(\ell)}(Y_{i,1},\dots,Y_{i,\ell})$ if the log likelihood ratio $\sum_{j=1}^\ell \log\left(\frac{P_1(Y_{i,j};\theta)}{P_0(Y_{i,j})}\right)$
is strictly monotonic increasing in $T_i^{(\ell)}$ for all $\theta \in {\Theta}$.
\end{definition}
In addition to being satisfied in the trivial case where $P_0$ and $P_1$ are fully specified, the monotone likelihood ratio assumption holds  in the conventional setting where  $P_0$ and $P_1$ belong to a common one-parameter family of distributions with a monotone likelihood function.  For example, in exponential families the test statistic $T_i^{(\ell)}$ is the sufficient statistic for $\theta$.
This property is well illustrated by the example of testing two Gaussian distributions discussed above.  If we assume the null distribution is known, but the larger mean of $P_1$ is unknown,  the procedure can still be implemented.  The sum of the measurements, $\sum_j Y_{i,j}$, is a sufficient statistic (whose distribution under $P_0$ of course does not depend on $P_1$).   For additional examples in which the underlying test satisfies the monotone likelihood ratio assumption, see \cite{karlin1956theory} and references therein.
}


\subsection{Details of Sequential Thresholding}
While the previous discussion highlighted the main principle behind Sequential Thresholding, the procedure becomes slightly more complicated in its full generality.  To show the procedure achieves the lower bound of Cor. \ref{STcor:LB} as $n$ grows large, both the allocation of measurements across steps and the proportion of null components discarded on each step must be adjusted. 


In general, Sequential Thresholding requires three inputs: \emph{1)} $\delta$, the desired family wise error rate, \emph{2)} a constant $\rho \in [1/2,1)$ representing the proportion of null components discarded on each step, and $\emph{3)}$ a total measurement budget $mn$ (meaning an average of $m$ samples per dimension).  
{Assume $P_0$, $P_1$, and a test statistic $T_i^{(\ell)}$ satisfy the monotone likelihood ratio assumption (Def. \ref{def:MLR}), and assume the procedure has exact knowledge of $s$ to facilitate analysis\footnote{In practice, the procedure is fairly insensitive to knowledge of $s$.  Specifically, when $s$ is underestimated, it is straightforward to see that while the procedure may exceed the measurement budget, it will also have a decreased family wise error rate.  Likewise, underestimating $s$ will result in a small increase in the family wise error rate, but also a decrease in the total expected number of samples.}.  The minimum expected proportion of null components discarded on each step, $\rho$, is fixed throughout the procedure and is used to define the series of thresholds as
\begin{eqnarray} \label{eqn:rho}
{
\min\left\{\gamma_k :  \P(T_i^{(m_k)} \leq \gamma_k |i \not \in \S ) \geq \rho\right\}.
}
\end{eqnarray}
In words, the thresholds are set so that at least a proportion $\rho$ of the null components are discarded, in expectation, on each step. 
Here, $m_k$ is the number of samples of any index measured on step $k$.  As $m_k$ is a function of the step index, so is the threshold $\gamma_k$. }

With $\rho$ and $\delta$ as inputs, and a total expected measurement budget $mn$, Sequential Thresholding operates as follows.  Let $\S_{k}$ denote the subset of $\{1,\dots,n\}$ comprised of components still under consideration at step $k$.  The procedure first initializes by setting $\S_1 = \{1,\dots,n\}$ and defining 
{
\begin{eqnarray}
K =  \left \lceil \log_\frac{1}{1-\rho} \left( \frac{2(n-s)}{\delta} \right) \right \rceil
\end{eqnarray}
where $\lceil x \rceil$ denotes smallest integer greater than or equal to $x$.
For steps $k = 1,\dots,K$, the procedure proceeds as follows.  On step $k$, each component in $\S_k$ is sampled $m_k$ times.  The number of samples taken on step $k$ is defined as  
\begin{eqnarray} \label{eqn:mk}
m_k =\left \lfloor m \: k \: \rho^2 \left( \frac{n}{n+sK^2}\right) \right \rfloor
\end{eqnarray}
where $\lfloor x \rfloor$ indicates the largest integer smaller than or equal to $x$.
The procedure then compares the test statistic comprised of the $m_k$ samples to the threshold defined in (\ref{eqn:rho}) and includes only the indices that exceed the threshold in the set of components to be sampled on the following step:
\begin{eqnarray} \nonumber
\mathcal{S}_{k+1} = \left\{ i \in \mathcal{S}_{k}: T_{i}^{(m_k)} > \gamma_k\right\}
\end{eqnarray}
where $\gamma_k$ is defined in (\ref{eqn:rho}).
In words, if $T_{i}^{(m_k)}$ is below $\gamma_k$, no further measurements of component $i$ are taken for the remainder of the procedure.  Otherwise, component $i$ is measured on the subsequent step.  By definition of $\gamma_k$, approximately $\rho$ times the number of remaining components following $P_0$ will be eliminated on each step; if $s \ll n$, each thresholding step eliminates approximately $\rho$ times the total number of components remaining.  After step $K$, the procedure terminates and estimates $\S$ as the indices still under consideration: $\hat{\S} = \S_{K+1}$.  The procedure is detailed in Alg. \ref{alg:sds}.
}


\begin{algorithm}[h]
\caption{Sequential Thresholding}
\begin{algorithmic} \label{alg:sds}
\STATE{input: desired error probability $\delta$, budget $m$, $\rho \in [1/2,1)$}
\STATE{initialize: $\S_1 = \{1,\dots,n\}$, $K =  \left \lceil \log_\frac{1}{1-\rho} \left( \frac{2(n-s)}{\delta} \right) \right \rceil$}
\FOR{$k = 1,\dots,K$}
\FOR{$ i \in \S_{k}$}
\STATE{{\bf measure}: {sample $m_k =\left \lfloor m \: k \: \rho^2 \left( \frac{n}{n+sK^2}\right) \right \rfloor$ times $Y_i$, \hspace*{.6cm} compute scalar statistic $T_i^{(m_k)}$ }}
\STATE{{\bf threshold}: {$\S_{k+1} \ := \ \{i\in\S_{k} \ : T_{i}^{(m_k)} > \gamma_k \}$ }}
\ENDFOR
\ENDFOR
\STATE{output: $\S_{K+1}$}
\end{algorithmic}
\end{algorithm}

\subsection{Ability of Sequential Thresholding}
For fixed $P_0$ and $P_1$ belonging to a class of distributions satisfying  Definition 5, the following theorem and corollary relate $(n,s,m)$ to the family wise error rate of the procedure.

\begin{thm} \emph{Finite sample performance of Sequential Thresholding}. \label{thm:STach}
Consider Sequential Thresholding defined in Alg. \ref{alg:sds}.   Provided 
\begin{eqnarray} \nonumber
m \geq  \frac{\log s + \log \delta^{-1} +\log 4 }{c_n} 
\end{eqnarray}
then 
\begin{eqnarray} \nonumber
\P_e \leq \delta
\end{eqnarray}
where  
\begin{eqnarray} \label{eqn:cn}
c_n &=&
 \rho^2 \left( \frac{n}{n+sK^2} \right) \times \\ \nonumber && 
 \left(D(P_0||P_1) - \sqrt{\frac{\sigma^2(P_0||P_1)}{\left( \frac{\rho^2  n \log s}{{D(P_0||P_1)} (n+sK^2)} -1\right)  \left(1-\rho\right)}}\right)
\end{eqnarray}
and is assumed to be positive.  
\end{thm}
\begin{proof}
See Appendix \ref{app:ST}.
\end{proof}

Theorem \ref{thm:STach} quantifies the expected number of samples per dimension in  the finite setting.  The theorem is in terms of a sequence, $c_n$, which, under certain conditions, approaches the Kullback-Leibler divergence between $P_0$ and $P_1$.  Proof of the theorem relies on techniques closely related to the Chernoff-Stein Lemma, and is found in the Appendix.

\begin{cor} \label{cor:ST} \emph{Reliability of Sequential Thresholding.}
If
\begin{eqnarray} \nonumber
  \lim_{n\rightarrow \infty} \frac{m}{\log s } > \frac{1}{D(P_0||P_1)}
\end{eqnarray}
then sequential threshold satisfies
\begin{eqnarray} \nonumber
\lim_{n \rightarrow \infty} \P_e = 0
\end{eqnarray}
with input parameters $\delta = \frac{1}{\log s}$ and $\rho = 1- \frac{1}{\sqrt{\log s}}$, provided $s < n/(\log n)^2$, $\lim_{n\rightarrow \infty} s = \infty$, and $\sigma^2(P_0||P_1) < \infty$.
\end{cor}
\begin{proof}
{The proof follows from Thm. \ref{thm:STach} by setting the input parameters as specified (which implies $ \lim_{n \rightarrow \infty} \delta = 0$). The total number of steps is then $K =  \left \lceil \left.  \log \left(2( n-s) \right) \right/ (2 \log \log s) -1/2  \right\rceil$.  With this $K$, and  with $\rho$ as defined in the statement of the theorem, $\lim_{n\rightarrow \infty } c_n = D(P_0||P_1)$, where $c_n$ is defined in (\ref{eqn:cn}). } Together with the forward part of the theorem, this implies the corollary.
\end{proof}

As Sequential Thresholding is uniform coordinate wise procedure, comparison of Cor. \ref{cor:ST} to Cor. \ref{STcor:LB} shows the procedure is asymptotically optimal in terms of the required number of samples needed for reliable recovery.

Thm. \ref{thm:STach} and  Cor. \ref{cor:ST}  imply that as the size of the problem increases (as $n$ goes to infinity), if $m$ is greater than $D(P_0||P_1)^{-1}\log s$, the procedure will succeed in exact recovery of the sparse support set.  This achieves the lower bound in Cor. \ref{STcor:LB}, which states that any reliable procedure requires at least $D(P_0||P_1)^{-1}\log s$ samples per dimension.   This implies that Sequential Thresholding is in a sense \emph{first order} optimal.  While not investigated here, one could also analyze the rate at which the procedure approaches the lower bound in Cor. \ref{STcor:LB}, although the authors suspect the procedure would not achieve \emph{second} order optimality {(as the rate at which $c_n$ approaches $D(P_0||P_1)^{-1}$ in Corollary \ref{cor:ST} is quite slow)}. {We suspect other procedures proposed for composite tests in the monotone likelihood setting  could be analyzed and/or modified to achieve higher order optimality in the sparse setting, in particular \cite{lai1988nearly}; we leave this for future work.}




This paper showed that sequential methods for support recovery of high dimensional sparse signals in noise can succeed using far fewer measurements than non-sequential methods. Specifically, non-sequential methods require the number of measurements to grow logarithmically with the dimension, while sequential methods succeed if the number of measurements grows logarithmically with the level of sparsity. A simple procedure termed Sequential Thresholding was shown to achieve the lower bound asymptotically. Sequential Thresholding can be implemented in the monotone likelihood setting, making it a practical solution for sparse recovery problems encountered in practice.

\bibliographystyle{IEEEtran}
\bibliography{limitsbibo}


\appendices

\section{Proof of Theorem \ref{thm:LBSeq}}

\label{app:LB}
\begin{proof}
{We first bound the family wise error rate in terms of the false positive and false negative probabilities associated with incorrectly assigning or excluding any element from $\hat{\mathcal{S}}$.    The coordinate-wise assumption implies that the individual error rates at each index are the same.  To be specific, define the event
 \begin{eqnarray} \label{eqn:E}
 \mathcal{E}_i = \{i \in \hat{\S} \Delta \S \} .
 \end{eqnarray}
 A coordinate wise procedure then has $\P({\cal{E}}_{i}) = \P({\cal{E}}_{i'})$ for all $i,i' \in \S$, and, likewise  $\P({\cal{E}}_{i}) = \P({\cal{E}}_{i'})$ for all $i,i' \not \in \S$.  Under this assumption we can simplify notation and define the false positive and false negative rates which are independent of the particular index: 
\begin{eqnarray} \label{eqn:alphabeta}
\alpha = \P({\cal{E}}_i | i \not \in \mathcal{S} )   \qquad \beta =  \P({\cal{E}}_i | i \in \mathcal{S} ).
\end{eqnarray} }
 From  (\ref{eqn:FWER1}),
\begin{eqnarray} \nonumber
\P_e &=& \P\left(\bigcup_{i \not \in \S} {\cal{E}}_i \cup \bigcup_{i \in \S} {\cal{E}}_i \right) \\ \nonumber
&=& 1 - \P\left(\bigcap_{i \not \in \S } {\cal{E}}_i^c \cap \bigcap_{i \in \S} {\cal{E}}_i^c\right) \\ 
 &=& 1 - (1-\beta)^s (1-\alpha)^{n-s} \nonumber \\  
&\geq& 1-e^{-\beta s} e^{-\alpha(n-s)} \label{eqn:Pelowbnd}
\end{eqnarray}
where the last inequality follows as $(1 - \beta)^n \leq e^{-\beta n}$ for $\beta \in [0,1]$ and $n \in \{1,\dots\}$.  To continue, we can bound the expected number of samples associated with any particular index.  From \cite{SeqAnalysis}, Thm. 2.39, the following holds for any binary hypothesis test 
\begin{eqnarray} \nonumber
m_0 \geq \frac{\alpha \log \left(\frac{\alpha}{1-\beta} \right) + (1-\alpha) \log \left(\frac{1-\alpha}{\beta} \right)  }{D(P_0||P_1)}
\end{eqnarray}
where $m_0$ is the expected number of samples of any component $i \not \in \mathcal{S}$.  We can further bound the expected number of samples as
\begin{eqnarray} \nonumber
m_0 & \geq & \frac{ \alpha \log \alpha +  (1-\alpha) \log ( 1-\alpha) +  (1-\alpha) \log {\beta}^{-1}} { D(P_0||P_1)}\\ \nonumber
& \geq & \frac{ (1-\alpha) \log \beta^{-1} - \log 2}{ D(P_0||P_1)}
\end{eqnarray}
where the first inequality follows as $\alpha \log (1/(1-\beta) )\geq 0$, and the last inequality follows as $\alpha \log \alpha + (1-\alpha) \log (1-\alpha) \geq \log(1/2)$, all $\alpha \in [0,1]$.  Likewise 
\begin{eqnarray} \nonumber
m_1 &\geq & \frac{(1-\beta) \log \left(\frac{1-\beta}{\alpha} \right) + \beta \log \left(\frac{\beta}{1-\alpha} \right) }{D(P_1||P_0)}\\ \nonumber
& \geq & \frac{ (1-\beta) \log \alpha^{-1} - \log 2}{ D(P_1||P_0)}
\end{eqnarray}
where $m_1$ is the average number of samples given $i \in \S$, and the first inequality is again from Thm. 2.39 of \cite{SeqAnalysis}.
Let $D_{\mathrm{KL}} = \max \left\{D(P_0||P_1),  D(P_1||P_0) \right\}$.  We have
\begin{eqnarray} \nonumber
m &=& \frac{(n-s) m_0 + s \; m_1}{n} \\ \nonumber
 &\geq & \frac{(n-s) (1-\alpha) \log \beta^{-1}  + s  (1-\beta) \log \alpha^{-1} - n \log 2}{n D_{\mathrm{KL}}}.
\end{eqnarray}
If $\alpha \leq \beta$ we have 
\begin{eqnarray} \nonumber
m &\geq & \frac{(n-s) (1-\beta) \log \beta^{-1}  + s  (1-\beta) \log \beta^{-1} - n \log 2}{n  D_{\mathrm{KL}}} \\ \nonumber
&=& \frac{(1-\beta) \log \beta^{-1} - \log 2}{D_{\mathrm{KL}}} \\ \nonumber
&\geq & \frac{ \log \left( \frac{1}{2\beta}\right) -  \log 2}{D_{\mathrm{KL}}}
\end{eqnarray}
where the last inequality is easily verified for $\beta \in [0,1]$.

Imposing the condition in the forward part of the theorem, $m \leq (\log s + \log  (4 \delta)^{-1} )/D_{\mathrm{KL}}$ gives
\begin{eqnarray} \nonumber
\frac{\log s   + \log \left(\frac{1}{4\delta}\right)}{D_{\mathrm{KL}}} \geq m \geq \frac{ \log \left( \frac{1}{2\beta}\right) - \log 2}{D_{\mathrm{KL}}}
\end{eqnarray}
which implies
\begin{eqnarray} \nonumber
 \log s + \log \left(\frac{1}{4\delta}\right) \geq  \log \left( \frac{1}{2\beta}\right)  - \log 2
\end{eqnarray}
and thus
\begin{eqnarray} \nonumber
\beta &\geq& {\delta}/{s}.
\end{eqnarray}
Hence, 
\begin{eqnarray} \nonumber
\P_e &\geq & 1 - e^{-\delta } e^{-(n-s) \alpha} \geq  1 - e^{-\delta  }.
\end{eqnarray}
Conversely if $\beta > \alpha$
\begin{eqnarray} \nonumber
m &\geq &\frac{\log \left( \frac{1}{2\alpha}\right) - \log 2}{D_{\mathrm{KL}}}
\end{eqnarray}
and 
\begin{eqnarray} \nonumber
\P_e &\geq & 1 - e^{-s \beta} e^{-\frac{ (n-s)  \delta }{s}  \alpha} 
\end{eqnarray}
which, provided $s < n/2$, gives
\begin{eqnarray} \nonumber
\P_e &\geq & 1 - e^{-\delta}.
\end{eqnarray}
completing the proof.  
\end{proof}

\section{Proof of Theorem \ref{thm:NSasmp}}
\label{app:NSMethods}
\begin{proof}
We write the family wise error rate as in (\ref{eqn:Pelowbnd}):
\begin{eqnarray} \nonumber
\P_e
&\geq& 1-e^{-\alpha(n-s)} e^{-\beta s} \\
&\geq  & 1-e^{-\alpha(n-s)} \nonumber
\end{eqnarray} {
where $\mathcal{E}_i$ is defined in (\ref{eqn:E})  and $\alpha$ and $\beta$ are defined in (\ref{eqn:alphabeta}).
Note that if $\alpha > \frac{1}{n-s}$, then $\P_e \geq 1-e^{-1} \geq 1/2$.  We have that
\begin{eqnarray} \nonumber
{\limsup_{m\rightarrow \infty} } \ \frac{1}{m} \log \alpha^{-1} < {\liminf_{m\rightarrow \infty} } \ \frac{1}{m} \log (n-s)
\end{eqnarray}
then ${\liminf_{n\rightarrow \infty}} \P_e \geq 1/2 $ (since the above inequality implies $\alpha > \frac{1}{n-s}$ for sufficiently large $m$). Next assume $\beta < 1/2$ (the result of the Theorem is trivial if $\beta \geq 1/2$).  From \cite[p. 386]{Cover:1991:EIT:129837} (Chernoff Information),
\begin{eqnarray}
\limsup_{m\rightarrow \infty} \frac{1}{m} \log \alpha^{-1} = D(P_{\lambda}||P_{0}) \leq  D(P_1||P_{0})
\end{eqnarray}
where
\begin{eqnarray} \nonumber
P_{\lambda} = \frac{P_0^\lambda P_1^{1-\lambda}}{\int_{\Omega} P_0^\lambda P_1^{1-\lambda}dy}
\end{eqnarray}
for $\lambda \in [0,1]$.  Thus, if 
\begin{eqnarray} \nonumber
\liminf_{m\rightarrow \infty} \frac{1}{m} \log (n-s) \geq D(P_1||P_{0})
\end{eqnarray}
 $\liminf_{n\rightarrow \infty} \P_e \geq 1/2 $.  Since $\lim_{n\rightarrow \infty} m = \infty$, this implies the result.  If
\begin{eqnarray} \label{eqn:asmpNS}
  \limsup_{n\rightarrow \infty} \frac{m}{\log n} < \frac{1}{D(P_1||P_0)}
\end{eqnarray}
then $\liminf_{n\rightarrow \infty}\P_e \geq 1/2$.}
\end{proof}

\section{Proof of Theorem \ref{thm:SPRTach}}
\label{app:SPRT}

\begin{proof}
For an SPRT with thresholds $\gL$ and $\gU$, from \cite{SeqAnalysis}, the following well known inequalities hold:
\begin{eqnarray}
\alpha \leq \gU^{-1} = \frac{1}{(n-s)^{1+\epsilon}} \qquad \qquad \beta \leq \gL = \frac{1}{s^{1+\epsilon}}
\end{eqnarray}
where $\alpha$ and $\beta$ are defined in (\ref{eqn:alphabeta}). From a union bound on the family-wise error rate
\begin{eqnarray}
\lim_{n\rightarrow \infty} \P_e \leq \lim_{n \rightarrow \infty} (n-s) \alpha + s \beta  = 0
\end{eqnarray}
implying the forward portion of the lemma.  

We can write the expected number of measurements per dimension as
\begin{eqnarray}  \nonumber
m &=& \frac{(n-s)\E_0\left[J_i\right] +s \;\E_1\left[J_i\right] }{n}
\end{eqnarray}
By Wald's identity \cite{SeqAnalysis}
\begin{eqnarray} \nonumber
\E_1\left[J_i\right] &=& \frac{\E_1\left[L_i^{(J_i)} \right]}{ \E_1\left[ L_i^{(1)} \right]} =  \frac{\E_1\left[L_i^{(J_i)} \right]}{ D(P_1||P_0)} 
\end{eqnarray}
and similarly,  $\E_0\left[J_i\right]  = \frac{-\E_0\left[L_i^{(J_i)} \right]}{ D(P_0||P_1)}$.
Dividing by $\log s$ and taking the limit, we have
\begin{eqnarray} \nonumber 
&& \hspace*{-.9cm} \lim_{n \rightarrow \infty} \frac{m}{\log s}  = \lim_{n \rightarrow \infty} \frac{(n-s)\E_0\left[J_i\right] +s \;\E_1\left[J_i\right] }{n \log s}\\  \nonumber
& = & \hspace{-.4cm} \lim_{n \rightarrow \infty} \frac{-(n-s)\E_0\left[L_i^{(J_i)} \right] }{n \log (s) D(P_0||P_1)} +\lim_{n\rightarrow \infty} \frac{ s \;\E_1\left[L_i^{(J_i)} \right] }{n \log (s) D(P_1||P_0)} \\
& \leq & \hspace{-.4cm} \lim_{n \rightarrow \infty} \frac{(n-s)  (\log \gL^{-1} + C_1 )    }{n \log (s) D(P_0||P_1)}
 + \lim_{n \rightarrow \infty} \frac{s   (\log \gU + C_2 )   }{n \log (s) D(P_0||P_1)} \nonumber
\end{eqnarray}
where the inequality follows by the assumptions in (\ref{STeqn:ass1}) and as 
\begin{eqnarray} \nonumber
-\E_0\left[L_i^{(J_i)} \right]  &=&   - (1-\alpha) \E_0\left[ \left. L_i^{(J_i)} \right \vert L_i^{(J_i)} < \log \gL  \right]    \\&& \hspace*{.5cm} -   \nonumber
\alpha \E_0\left[ \left. L_i^{(J_i)} \right \vert L_i^{(J_i)} > \gU \right]  \\ \nonumber
 &\leq &  (1-\alpha)  (\log \gL^{-1} + C_1 ) + \alpha(\log \gU^{-1} - C_2) \\ \nonumber
 &\leq & (1-\alpha)  (\log \gL^{-1} + C_1 ) \\ \nonumber
 &\leq& \log \gL^{-1} + C_1 
\end{eqnarray}
and likewise,
\begin{eqnarray} \nonumber
\E_1\left[L_i^{(J_i)} \right] \leq \log \gU + C_2 . 
\end{eqnarray}
Using the prescribed values of $\gU$ and $\gL$ gives 
\begin{eqnarray}
\lim_{n \rightarrow \infty} \frac{m}{\log s} = \frac{1+\epsilon}{D(P_0||P_1)}
\end{eqnarray}
provided $s< n/ \log n$, completing the proof.
\end{proof}

\section{Proof of Theorem \ref{thm:STach}}
\label{app:ST}

\begin{proof}
{By design, Sequential Thresholding satisfies the measurement budget in (\ref{eqn:budget}). Consider the expected number of samples required by Sequential Thresholding:
\begin{eqnarray} \label{eqn:budgetEq}
  \E\left[\sum_{i,j} \Gamma_{i,j} \right]  & =&   \mathbb{P}\left(\bigcup_{i \in \mathcal{S}} \mathcal{E}_i \right) \E\left[\sum_{i,j} \Gamma_{i,j} \left\vert   \bigcup_{i \in \mathcal{S}} \mathcal{E}_i \right. \right]   \\
     && \hspace*{-1.8cm} +\; \mathbb{P}\left( \bigcap_{i  \in \mathcal{S}} \mathcal{E}_i^c \right)\E\left[\sum_{i,j} \Gamma_{i,j} \left\vert  \bigcap_{i  \in \mathcal{S}} \mathcal{E}_i^c \right.  \right] \nonumber
\end{eqnarray}
where the equality follows from the law of total probability and conditioning on one or more false negative events.  From Alg. \ref{alg:sds}, one or more false negatives can only reduce the total number of samples, and we have 
\begin{eqnarray} \nonumber
\E\left[\sum_{i,j} \Gamma_{i,j}  \left\vert   \bigcap_{i \in \mathcal{S}} \mathcal{E}_i^c \right. \right]  \geq \E\left[\sum_{i,j} \Gamma_{i,j} \left\vert   \bigcup_{i \in \mathcal{S}} \mathcal{E}_i \right. \right].
\end{eqnarray}
Combining this with (\ref{eqn:budgetEq}) gives
\begin{eqnarray}      \nonumber
 && \hspace*{-1cm} \E\left[\sum_{i,j} \Gamma_{i,j} \right]  \leq   \E\left[\sum_{i,j} \Gamma_{i,j} \left\vert   \bigcap_{i \in \mathcal{S}} \mathcal{E}_i^c \right.  \right] \\  \nonumber
  = &&  \hspace*{-.2cm}  \sum_{k=1}^{K} m_k \left((1-\rho)^{k-1} (n-s) + s \right) \\ \nonumber
   \leq && \hspace*{-.2cm} \sum_{k=1}^{K} m \left(\frac{n}{n+s K^2}\right)   k  \rho^2 \left((1-\rho)^{k-1} (n-s) + s \right) \\ \nonumber
  \leq & & \hspace*{-.2cm}  mn \left(\frac{n-s+sK^2}{n+s K^2}\right) \leq  mn.
\end{eqnarray}
Here, the equality follows from independence of the samples across the $K$ steps.  The third inequality follows as the sum is an arithmetico-geometric series for $\rho \in [1/2, 1)$ and as $\sum_{k=1}^K k \leq K^2$. As the budget is satisfied, this implies the procedure will use less than $m$ samples per dimension. We continue by bounding the error rates.}

\begin{figure*}[t!]
\normalsize
\newcounter{MYtempeqncnt}
\setcounter{MYtempeqncnt}{\value{equation}}
\setcounter{equation}{32}
\begin{eqnarray} \label{eqn:long1}
\beta &\leq& \sum_{k=1}^K \mathbb{P}_1 \left( L_{i}^{(m_k)} \leq \gamma_k  \right) \leq  \sum_{k=1}^K e^{-m_k(D(P_0||P_1) - \epsilon_k)  } \\ \label{eqn:cnprime}
&\leq& \sum_{k=1}^K \exp \left(-m k \underbrace{ \rho^2 \left( \frac{n}{n+sK^2} \right) \left(D(P_0||P_1) - \sqrt{\frac{\sigma^2(P_0||P_1)}{\left( \frac{m \rho^2  n}{n+sK^2} -1\right)  \left(1-\rho\right)}}\right)}_{c_n'}  \right)
\leq  \frac{e^{-mc_n'}}{ 1 - e^{-mc_n'}} \label{eqn:betabnd1}
\end{eqnarray}
\hrulefill
\begin{eqnarray} \label{eqn:cnlong}
c_n = \rho^2 \left( \frac{n}{n+sK^2} \right) \left(D(P_0||P_1) - \sqrt{\frac{\sigma^2(P_0||P_1)}{\left( \frac{\rho^2  n \log s}{{D(P_0||P_1)} (n+sK^2)} -1\right)  \left(1-\rho\right)}}\right)
\end{eqnarray}
\setcounter{equation}{\value{MYtempeqncnt}}
\addtocounter{equation}{0}
\hrulefill
\vspace*{4pt}
\end{figure*}

Applying a union bound to the family wise error rate, we have
\begin{eqnarray} \label{eqn:fwer_st}
  \P_e &\leq& (n-s) \alpha + s \beta
\end{eqnarray}
where $\alpha$ and $\beta$ are defined in (\ref{eqn:alphabeta}).
The false negative event is given as
\begin{eqnarray} \nonumber
\beta &=& \P_1 \left(\bigcup_{k=1}^K  L_{i}^{(m_k)} < \gamma_k \right) \\
&\leq& \sum_{k=1}^K \P_1 \left( L_{i}^{(m_k)} < \gamma_k \right) \nonumber
\end{eqnarray}
{where we define $\P_1(\cdot) = \P(\cdot|i \in \S)$ to simplify notation.
We continue by bounding the above probability. The following analysis is closely related to the Chernoff-Stein lemma \cite{Cover:1991:EIT:129837}, but modified for one sided tests.   Note by the mononicity of $T_i^{(m_k)}$ with respect to $L_i^{(m_k)}$, we can analysis the test using $L_i^{(m_k)}$.}  Let $\bm{y}_k = (y_{1},\dots,y_{m_k})$ and define the region $\mathcal{A}_k \subset \mathbb{R}^{m_k}$ as
\begin{eqnarray} \nonumber
\mathcal{A}_k:=\{\y_k: L_i^{(m_k)}(\bm{y}_k)  < \gamma_k \}. 
\end{eqnarray}
For all $\y_k$ in $\mathcal{A}_k$, by definition, 
\begin{eqnarray} \nonumber
L_i^{(m_k)}(\bm{y}_k) = \sum_{j=1}^{m_k} \log\frac{P_1(y_j)}{P_0(y_j) } < \gamma_k
\end{eqnarray}
which implies
\begin{eqnarray} \nonumber
P^{{(m_k)}}_1(\y_k)  < e^{\gamma_k} P^{(m_k)}_0(\y_k)
\end{eqnarray}
where $P^{(m_k)}_1(\y_k)  =  \prod_{j=1}^{m_k} P_1(y_j)$.  Again by definition 
{
\begin{eqnarray} \nonumber
\mathbb{P}_1(L_i^{(m_k)} < \gamma_k) &=& \int_{\mathcal{A}} P_1^{(m_k)}(\y) d\y \\  \nonumber
&\leq& \int_{\mathcal{A}} e^{\gamma_k} P_0^{(m_k)}(\y) d\y \\ \nonumber
&=& e^{\gamma_k}  \int_{\mathcal{A}} P_0^{(m_k)}(\y) d\y \\ \nonumber
&=& e^{\gamma_k}  \mathbb{P}_0(L_i^{(m_k)}< \gamma_k ) \\ \label{eqn:P1bnd1}
&\leq & e^{\gamma_k}.
\end{eqnarray}
The above relationship holds for any $\gamma_k$, including that defined in (\ref{eqn:rho}).  Next define $\epsilon_k >0$ such that
\begin{eqnarray} \label{eqn:thesDef}
\gamma_k = -m_k(D(P_0||P_1) - \epsilon_k)
\end{eqnarray}
where $\gamma_k$ is given in (\ref{eqn:rho}). 
It remains to show what values of $\epsilon_k$ simultaneously satisfy (\ref{eqn:rho}) for any $\rho \in [1/2,1)$.  Specifically, we need to find the range of values of $\epsilon_k$ that satisfy
\begin{eqnarray} \label{eqn:P0L}
\P_0 \left( L_i^{(m_k)} \leq  -m_k(D(P_0||P_1) - \epsilon_k) \right) \geq \rho.
\end{eqnarray} 
Proceeding, 
\begin{eqnarray} \nonumber
& & \hspace*{-1.4cm} \P_0 \left( L_i^{(m_k)} \leq  -m_k(D(P_0||P_1) - \epsilon_k) \right) \\ \nonumber
&=& \P_0\left(\frac{1}{m_k} L_i^{(m_k)} + D(P_0||P_1) \leq \epsilon_k \right) \\  \nonumber
&=& \P_0\left(D(P_0||P_1) - \frac{1}{m_k}\sum_{j=1}^{ m_k} \log\frac{P_0(Y_j)}{P_1(Y_j)} \leq \epsilon_k\right) \\  \label{eqn:KLconverST} 
& \geq & 1 - \frac{\sigma^2(P_0||P_1)}{m_k \epsilon_k^2} 
\end{eqnarray}
where the last line follows from Chebyshev's inequality.  To insure  that (\ref{eqn:P0L}) can be satisfied, we have the following condition
\begin{eqnarray} \nonumber
\epsilon_k \geq \sqrt{\frac{\sigma^2(P_0||P_1)}{m_k (1-\rho)}}.
\end{eqnarray}
As $m_k$ is smallest for $k = 1$, from the definition of $m_k$ in (\ref{eqn:mk}),
\begin{eqnarray} \nonumber
m_k \geq m_1 \geq  \frac{m \rho^2 n}{n+sK^2} - 1
\end{eqnarray}
and the condition can be satisfied for any $k$ provided
\begin{eqnarray} \label{eqn:ekReq}
\epsilon_k \geq \sqrt{\frac{\sigma^2(P_0||P_1)}{\left( \frac{m \rho^2  n}{n+sK^2} -1\right)  \left(1-\rho\right)}}.
\end{eqnarray}
}

To summarize developments thus far, we've shown that if $\gamma_k = -m_k (D(P_0||P_1) - \epsilon_k) $ then both
\begin{eqnarray} \nonumber
\P_1 \left( L_i^{(m_k)} < \gamma_k  \right) \leq  e^{-m_k(D(P_0||P_1) - \epsilon_k)  }
\end{eqnarray}
and 
\begin{eqnarray} \nonumber
\P_0 \left( L_i^{(m_k)} \leq \gamma_k \right)  \geq \rho
\end{eqnarray}
for any $\epsilon_k$ that satisfies (\ref{eqn:ekReq}).

\addtocounter{equation}{3}

Continuing with $m_k$ as specified in (\ref{eqn:mk}), gives (\ref{eqn:long1}) -- (\ref{eqn:cnprime}), where the last inequality follows as the sum is geometric.   
With  $K =  \left \lceil \log_\frac{1}{1-\rho} \left( \frac{2(n-s)}{\delta} \right) \right \rceil$, the false positive rate is then
\begin{eqnarray} \nonumber
\alpha &\leq & (1-\rho)^K \\ \nonumber
&\leq & (1-\rho)^{\log_\frac{1}{1-\rho} \left( \frac{2(n-s)}{\delta} \right) } \\
&\leq & \frac{\delta}{2(n-s)} \label{eqn:alphabnd1}.
\end{eqnarray}
Combining (\ref{eqn:betabnd1}) and (\ref{eqn:alphabnd1}) gives
\begin{eqnarray} \nonumber
\P_e \leq  \frac{\delta}{2} + s \; \frac{e^{-m c_n'}}{ 1 - e^{-mc_n'}}.
\end{eqnarray}
Next, from the statement of the theorem, let 
\begin{eqnarray} \nonumber
m \geq \frac{\log \left( \frac{4s}{\delta}\right)}{c_n}
\end{eqnarray}
where $c_n$ is defined in (\ref{eqn:cnlong}).
Notice $c_n \leq D(P_0||P_1)$.  Thus, $m \geq \frac{\log s}{D(P_0||P_1)}$, and
 \begin{eqnarray} \nonumber
m \geq \frac{\log \left( \frac{4s}{\delta}\right)}{c_n} \geq \frac{\log \left( \frac{4s}{\delta}\right)}{c_n'} 
\end{eqnarray} 
where $c_n'$ is given in (\ref{eqn:cnprime}).
As $m \geq \frac{\log \left( \frac{4s}{\delta}\right)}{c_n'}$,
\begin{eqnarray} \nonumber
\P_e &\leq & \frac{\delta}{2} +  \frac{\delta/4}{ 1 - \delta/(4s)}  \leq  \delta
\end{eqnarray}
where the last inequality holds for $s \geq 1$, $\delta \leq 1$ which proves the theorem.

\end{proof}

\begin{IEEEbiographynophoto}{Matthew L. Malloy} received the B.S. degree in electrical and computer engineering from the University of Wisconsin in 2004, the M.S. degree from Stanford University in 2005 in electrical engineering, and the Ph.D. degree from the University of Wisconsin in 2012 in electrical and computer engineering.   

Dr. Malloy is currently a Principle Data Scientist at ComScore, Inc.  Previously, Dr. Malloy was postdoctoral researcher at the University of Wisconsin in the Wisconsin Institutes for Discovery from 2012 - 2013.  From 2005-2008, he was a signal processing and radio frequency design engineering for Motorola in Arlington Heights, IL, USA.  In 2008 he received the Wisconsin Distinguished Graduate fellowship.  In 2009, 2010 and 2011 he received the Innovative Signal Analysis fellowship.  

Dr. Malloy has served as a reviewer for the IEEE Transactions on Signal Processing, the IEEE Transactions on Information Theory, IEEE Transactions on Automatic Control, the IEEE Journal of Oceanic Engineering and the Annals of Statistics, and has served on the technical program committee for the IEEE Global Conference on Signal and Information Processing.  He was awarded first place in the best student paper contest at the Asilomar Conference on Signals and Systems in 2011.   His research interests include signal processing, estimation and detection, information theory and optimization, with a broad range of applications including data science, communications and biology.
\end{IEEEbiographynophoto}

\begin{IEEEbiographynophoto}{Robert D. Nowak} received the B.S., M.S., and Ph.D. degrees in electrical engineering from the University of Wisconsin-Madison in 1990, 1992, and 1995, respectively.  He was a Postdoctoral Fellow at Rice University in 1995-1996, an Assistant Professor at Michigan State University from 1996-1999,  held Assistant and Associate Professor positions at Rice University from 1999-2003, and is now the McFarland-Bascom Professor of Engineering at the University of Wisconsin-Madison. 

Professor Nowak has held visiting positions at INRIA, Sophia-Antipolis (2001), and Trinity College, Cambridge (2010). He has served as an Associate Editor for the IEEE Transactions on Image Processing and the ACM Transactions on Sensor Networks, and as the Secretary of the SIAM Activity Group on Imaging Science. He was General Chair for the 2007 IEEE Statistical Signal Processing workshop and Technical Program Chair for the 2003 IEEE Statistical Signal Processing Workshop and the 2004 IEEE/ACM International Symposium on Information Processing in Sensor Networks. 

Professor Nowak received the General Electric Genius of Invention Award (1993), the National Science Foundation CAREER Award (1997), the Army Research Office Young Investigator Program  Award (1999), the Office of Naval Research Young Investigator Program Award (2000), the IEEE Signal Processing Society Young Author Best Paper Award (2000), the IEEE Signal Processing Society Best Paper Award (2011), and the ASPRS Talbert Abrams Paper Award (2012). He is a Fellow of the Institute of Electrical and Electronics Engineers (IEEE). His research interests include signal processing, machine learning, imaging and network science, and applications in communications, bioimaging, and systems biology.
\end{IEEEbiographynophoto}

\end{document}